\documentclass[letterpaper, 10 pt, conference]{ieeeconf}  

\IEEEoverridecommandlockouts                              

\usepackage{amsmath}
\usepackage{amsfonts}
\newtheorem{theorem}{Theorem}
\newtheorem{corollary}{Corollary}
\newtheorem{lemma}{Lemma}
\newtheorem{definition}{Definition}
\newtheorem{assumption}{Assumption}

\newtheorem{proposition}{Proposition}
\usepackage{listings}
\usepackage{color}
\usepackage{array}
\usepackage{tabu}
\usepackage{cite}
\usepackage{mathrsfs,amssymb}
\usepackage{graphicx}
\graphicspath{ {image22/} }
\usepackage{float}
\usepackage{algorithm}
\usepackage{algpseudocode}
\algdef{SE}[DOWHILE]{Do}{doWhile}{\algorithmicdo}[1]{\algorithmicwhile\ #1}%
\DeclareMathOperator*{\argmin}{arg\,min}
\DeclareMathOperator*{\argmax}{arg\,max}
\DeclareMathOperator*{\card}{card}
\DeclareMathOperator*{\supp}{supp}

\title{\LARGE \bf
A Robust Circle-criterion Observer-based Estimator for Discrete-time Nonlinear Systems in the Presence of Sensor Attacks and Measurement Noise
}

\author{Tianci Yang, Carlos Murguia, Margreta Kuijper, Dragan  Ne\v{s}i\'{c} 
\thanks{This work was supported by the Australian Research Council under the Discovery Project DP170104099.}
\thanks{The authors are with the Department of Electrical and Electronics Engineering, the University of Melbourne, Australia.
        {\tt\small tianciy@student.unimelb.edu.au}}
}

\begin{document}

\maketitle
\thispagestyle{empty}
\pagestyle{empty}

\begin{abstract}
	We address the problem of robust state estimation and attack isolation for a class of discrete-time nonlinear systems with positive-slope nonlinearities under (potentially unbounded) sensor attacks and measurement noise. We consider the case when a subset of sensors is subject to additive false data injection attacks. Using a bank of circle-criterion observers, each observer leading to an Input-to-State Stable (ISS) estimation error, we propose a estimator that provides robust estimates of the system state in spite of sensor attacks and measurement noise; and an algorithm for detecting and isolating sensor attacks. Our results make use of the ISS property of the observers to check whether the trajectories of observers are “consistent” with the attack-free trajectories of the system. Simulations results are presented to illustrate the performance of the results.
\end{abstract}

\section{Introduction}

Networked Control Systems (NCSs) have emerged as a technology that combines control, communication, and computation and offers the necessary flexibility to meet new demands in distributed and large scale systems. Recently, security of NCSs has become an important issue as wireless communication networks might serve as new access points for attackers to adversely affect the operation of the system dynamics. Cyber-physical attacks on NCSs have caused substantial damage to a number of physical processes. One of the most well-known examples is the attack on Maroochy Shire Council’s sewage control system in Queensland, Australia that happened in January 2000. The attacker hacked into the controllers that activate and deactivate valves and caused flooding of the grounds of a hotel, a park, and a river with a million liters of sewage. Another incident is the very recent SuxNet virus that targeted Siemens’ supervisory control and data acquisition systems which are used in many industrial processes. It follows that strategic mechanisms to identify and deal with attacks on NCSs are strongly needed.

In \cite{Fawzi2012}-\nocite{Massoumnia1986}\nocite{Pajic2014}\nocite{Mo2014}\nocite{Vamvoudakis2014}\nocite{Chong2016b}\nocite{Vamvoudakis2012}\nocite{Shoukry2014}\nocite{Yong}\nocite{Park2015}\nocite{Liu2009}\nocite{Teixeira2012b}\nocite{Murguia2016}\nocite{Dolk1}\nocite{Hashemil2017}\nocite{Pasqualetti123}\nocite{Murguia2017d}\nocite{Jairo}\nocite{Carlos_Justin2}\nocite{Sahand2017}\nocite{Carlos_Justin3}\cite{Carlos_Iman1}, a range of topics related to security of control systems have been discussed. In general, they provide analysis tools for quantifying the performance degradation induced by different classes of attacks; and propose reaction strategies to identify and counter their effect on the system dynamics. Most of the existing work, however, has considered control systems with linear dynamics, although
in most engineering applications the dynamics of the plants being monitored and controlled is highly nonlinear. There are some results addressing the nonlinear case though. In \cite{Kim2016}, exploiting sensor redundancy, the authors address the problem of sensor attack detection and state estimation for uniformly observable continuous-time nonlinear systems. Similarly, in \cite{tianci1}, the authors provide an algorithm for isolating sensor attacks for a class of discrete-time nonlinear systems with bounded measurement noise.

In this manuscript,
we consider the
case when the system has $p$ sensors, all of which are
subject to measurement noise and up to $q < p/2$ of them are
attacked. Following the results in \cite{Chong2015} for linear systems, using a bank of circle criterion observers \cite{Arcak}\nocite{Fan2003}\nocite{Ibrir2007}-\cite{Sundaram2016}, each observer leading to an ISS estimation error, we construct an estimator that provides robust estimates of the system state in spite of sensor attacks. In particular, the proposed estimator leads to estimation errors satisfying an ISS property with respect to measurement noise but independent of attack signals. Next, we propose an algorithm for detecting and isolating false data injection sensor attacks. Our results make use of the ISS property of the observers to check whether the trajectories of observers are “consistent” with the attack-free trajectories of the system. The main idea behind our results is the following. Each observer in the bank is driven by a different subset of sensors. Thus, without attacks, the observers produce ISS estimation errors with respect to measurement noise only. For every pair of observers in the bank, we compute the largest difference between their estimates. If a pair of observers is driven by a subset of attack-free sensors, then the largest difference between their estimates is also ISS with respect to measurement noise only. However, if there are attacks on some of the sensors, the observers driven by those sensors might produce larger differences than the attack-free ones. These ideas work well under the assumption that less than $p/2$ sensors are attacked, i.e, $q < p/2$. To design the observers in the bank, we give an extension to the result in \cite{Ibrir2007} for designing robust discrete-time circle-criterion observers. In particular, we use the incremental multiplier technique introduced in \cite{Sundaram2016} to cast the observer design as the solution of a semidefinite program. We minimize the ISS-gain from the measurement noise to the estimation error.

The paper is organized as follows. In Section \uppercase\expandafter{\romannumeral2}, we present preliminary results needed for the subsequent sections. In Section \uppercase\expandafter{\romannumeral3}, we provide tools for designing optimal robust circle criterion observers in the attack-free case. In Section \uppercase\expandafter{\romannumeral4},
assuming that a sufficiently small number of sensors are subject to attacks, we propose an estimation scheme using a bank of robust circle criterion observers. In Section \uppercase\expandafter{\romannumeral5}, an algorithm for isolating sensor attacks is given. Finally, in Section \uppercase\expandafter{\romannumeral6}, we give concluding remarks.
\section{Preliminaries}
\subsection{Notation}

We denote the set of real numbers by $\mathbb{R}$, the set of natural numbers by $\mathbb{N}$ , the set of integers by $\mathbb{Z}$, and $\mathbb{R}^{n\times m}$ the set of $n\times m$ matrices for any $m,n \in \mathbb{N}$. For any vector $v\in\mathbb{R}^{n}$,  we denote {$v_{J}$} the stacking of all $v_{i}$, $i\in J$ and $J\subset \left\lbrace 1,\cdots,n\right\rbrace$, $|v|=\sqrt{v^{\top} v}$ and $\supp(v)=\left\lbrace i\in\left\lbrace 1,\cdots,n\right\rbrace |v_{i}\neq0\right\rbrace $. For a sequence of vectors $\left\lbrace v(k)\right\rbrace _{k=0}^{\infty}$, we denote $v_{[0,k]}$ a sequence of vectors $v(i)$, $i=0,\cdots,k$,  $||v||_{\infty}\triangleq \sup_{k\geq 0}|v(k)|$ and $||v||_{T}\triangleq \sup_{0\leq k\leq T}|v(k)|$. We say a sequence $\left\lbrace v(k)\right\rbrace \in l_{\infty}$ if $||v||_{\infty}<\infty$. We denote the cardinality of a set $S$ as $\card(S)$. We denote matrix $P$ to be positive definite as $P>0$. The identity matrix is denoted by $I$. A function $\beta:\mathbb{R}_{\geq0}\times\mathbb{R}_{\geq 0}\to\mathbb{R}_{\geq0}$ is said to be of class $exp-KL$ if there exist $c>0$, $\lambda\in(0,1)$, such that $\beta(s,k)=c\lambda^{k}\cdot s$. The binomial coefficient is denoted as $\binom{a}{b}$, where $a,b$ are nonnegative integers. We denote a variable $m$ uniformly distributed in the interval $(a,b)$ as $m\sim\mathcal{U}(a,b)$. 
\subsection{Definitions and lemmas}

Several definitions and lemmas that are important in this paper are introduced here.

\begin{definition}\cite{Sundaram2016}\label{222}
	(Incremental Multiplier Matrices). Suppose $f: \mathbb{R}^{n_{q}}\to \mathbb{R}^{n_{f}}$. A symmetric matrix $M\in\mathbb{R}^{(n_{q}+n_{f})\times(n_{q}+n_{f})}$ is an incremental multiplier matrix $(\delta MM)$ for $f$ if the following incremental quadratic constraint $(\delta QC)$ is satisfied for all $q_{1},q_{2}\in\mathbb{R}^{n_{q}}$:
	\begin{eqnarray}
	\left[ \begin{matrix}
	\triangle q\\
	\triangle f
	\end{matrix}\right] ^{T}M	\left[ \begin{matrix}
	\triangle q\\
	\triangle f
	\end{matrix}\right]\geq 0\label{20},
	\end{eqnarray}
	where $\triangle q=q_{1}-q_{2}$ and $\triangle f=f(q_{1})-f(q_{2})$. 
\end{definition}

\begin{definition}
	Consider a discrete-time system
	\begin{equation}
	e^{+}=F(e,m)\label{1},
	\end{equation}
	with state $e\in\mathbb{R}^{n}$, the input $m\in\mathbb{R}^{p}$ with $\left\lbrace m(k)\right\rbrace \in l_{\infty}$. The system is said to be input-to-state stable (ISS) with a linear gain $\gamma$ and an $exp-KL$ function if there exist $c>0$, $\lambda\in(0,1)$, and $\gamma\geq 0$ such that the following condition is satisfied:
	\begin{equation}
	|e(k)|\leq c\lambda^{k}|e(0)|+\gamma||m||_{k}\label{92},
	\end{equation}
	for all $e(0)\in\mathbb{R}^{n}$, $k\geq 0$, and $\left\lbrace m(k)\right\rbrace\in l_{\infty}$. 
\end{definition}
The next lemma provides sufficient conditions for system (\ref{1}) to be ISS with a linear gain $\gamma$ and an $exp-KL$ function. It is a special case of the lemma in \cite{Sundaram2016}.


\begin{lemma}\label{888} \cite{Sundaram2016} Let $c_{1},c_{2}>0$, $c_{3}\in(0,1)$ and $\mu_{1}>0$. Suppose there exists $V: \mathbb{R}^{n}\to\mathbb{R}_{\geq 0} $ such that
		\begin{equation}
		c_{1}|e|^{2}\leq V(e)\leq c_{2}|e|^{2}\label{8},
		\end{equation}
		\begin{equation}
		V(F(e,m))-V(e)\leq -c_{3} V(e)+c_{3}\mu_{1}|m|^{2}\label{22},
		\end{equation}
		for all $k\geq 0$, $e\in \mathbb{R}^{n}$, $m\in\mathbb{R}^{p}$. Then the system (\ref{1}) is ISS with a linear gain and an $exp-KL$ function with respect to the bounded sequence $\left\lbrace m(k)\right\rbrace $, and
		\begin{equation}\label{c}
		|e(k)|\leq c\lambda^{k}|e(0)|+\gamma||m||_{k},
		\end{equation} 
		for all $k\geq 0$, $e(0)\in\mathbb{R}^{n}$, and $\left\lbrace m(k)\right\rbrace \in l_{\infty}$, where $c=\sqrt{\frac{c_{2}}{c_{1}}}$, $\lambda=\sqrt{1-c_{3}}$, and $\gamma=\sqrt{\frac{\mu_{1}}{c_{1}}}$.
\end{lemma}
\section{A circle-criterion observer robust to measurement noise}

In \cite{Ibrir2007}, Ibrir designs a discrete-time nonlinear observers through circle criterion, but no disturbance in the system is considered. Our goal in this section is to give an extension of the result given in \cite{Ibrir2007} by taking measurement noise into consideration and propose a design method of a robust circle-criterion observer with respect to measurement noise in the absence of attack signals. The design method uses a similar idea as the one in \cite{Sundaram2016} by characterizing the nonlinearity with an incremental multiplier matrix, but we present an extension of the results given in \cite{Sundaram2016} by solving an optimization problem with more degrees of freedom, which leads to a less conservative ISS gain. More specifically, we show that in some circumstances our circle-criterion observer provides state estimates less sensitive to measurement noise (see Example 1). We consider a discrete-time nonlinear system formulated as:
\begin{eqnarray}
\begin{split}
x^{+}=&Ax+Gf(Hx)+\rho(u,y)\label{13},\\
y=&Cx+m,
\end{split}
\end{eqnarray}
where $x\in \mathbb{R}^{n}$ is the state, $y\in \mathbb{R}^{n_{y}}$ is the sensor measurements, $m\in \mathbb{R}^{n_{y}}$ is the measurement noise with $\left\lbrace m(k)\right\rbrace \in l_{\infty}$ and $G\in\mathbb{R}^{n\times r}$, $H\in\mathbb{R}^{r\times n}$.
The term $\rho(u,y)$ is a known arbitrary real-valued vector that depends on the system inputs and outputs.

The state-dependent nonlinearity $f(Hx)$ is an $r$-dimensional vector where each entry is a function of a linear combination of the states
\begin{equation}\label{230}
f_{i}=f_{i}\left( \sum_{j=1}^{n}H_{ij}x_{j}\right) ,\quad i=1,\cdots,r
\end{equation} 
where $H_{ij}$ are the entries of matrix $H$.
\vspace{2mm}
\begin{assumption}\label{1000} For any $i\in\left\lbrace 1,\cdots,r\right\rbrace $, the following holds,
	\begin{equation}
	\frac{f_{i}(v_{i})-f_{i}(w_{i})}{v_{i}-w_{i}}\geq 0,\forall v_{i},w_{i}\in\mathbb{R}\quad with\quad v_{i}\neq w_{i}\label{14}
	\end{equation}
\end{assumption}
\vspace{2mm}
We consider a circle-criterion observer with the following structure:
\begin{equation}
\hat{x}^{+}=A\hat{x}+Gf(H\hat{x}+K(C\hat{x}-y))+L(C\hat{x}-y)+\rho(u,y)\label{97}
\end{equation}
where $\hat{x}$ denotes the estimate of the state $x$, and $K\in\mathbb{R}^{r\times n_{y}}$, and $L\in\mathbb{R}^{n\times n_{y}}$ are the observer gains to be designed. Then the error $e=\hat{x}-x$ has the following dynamics
\begin{equation}
e^{+}=(A+LC)e-Lm+G\triangle f\label{5},
\end{equation}
where
\begin{equation}
\triangle f=f(\hat{q})-f(\tilde{q})\label{9},
\end{equation}
where $\tilde{q}=Hx$ and $\hat{q}=H\hat{x}+K(\hat{y}-y)$ and $\hat{y}=C\hat{x}$
\begin{equation}
\triangle q=\hat{q}-\tilde{q}=(H+KC)e-Km\label{19}.
\end{equation}

Our objective is to design the gains $K$ and $L$, such that a quadratic Lyapunov function $V(e)$ satisfies (\ref{8}) and (\ref{22}). Then we can show that the error dynamics of the observer is ISS with a linear gain and an $exp-KL$ function with respect to the bounded measurement noise and (\ref{c}) holds.


\begin{proposition}\label{234} Consider system (\ref{13}), for given $c_{3}\in(0,1)$, suppose there exist matrix $P\in\mathbb{R}^{n\times n}$ and $P>0$, $K\in\mathbb{R}^{r\times n_{y}}$ and $Y\in\mathbb{R}^{n\times n_{y}}$, an incremental multiplier matrix $M$ for the nonlinearity $f$, and scalars $\mu>0$ and $\mu_{1}>0$ that satisfy the matrix inequalities:
	\begin{eqnarray}
	\begin{split}
	\left[ \begin{matrix}
	-P&\star\\
	\Xi_{21}&\Xi_{22}
	\end{matrix}\right]+\left[\begin{matrix}
	0&0\\
	0&\Gamma^{T}M\Gamma
	\end{matrix} \right] \leq& 0\label{29},\\
	\left[ \begin{matrix}P&I\\
	I&\mu I\end{matrix}\right] \geq& 0,
	\end{split}
	\end{eqnarray} 
	where 
	\begin{eqnarray}
	\begin{split}
		\Xi_{21}^{T}=&\left[ \begin{matrix}
	PA+YC&-Y&PG
	\end{matrix}\right] \label{225},\\
	\Xi_{22}=&\left[ \begin{matrix}
	(c_{3}-1)P&0&0\\0&-c_{3}\mu_{1}I&0\\
	0&0&0
	\end{matrix}\right],
	\end{split}
	\end{eqnarray}
	and
	\begin{equation}
	\Gamma=\left[ \begin{matrix}
	H+KC&-K&0\\
	0&0&I
	\end{matrix}\right] \label{39},
	\end{equation}
	then the observer (\ref{97}) characterized by gains $L=P^{-1}Y$ and $K$ has ISS error dynamics with a linear gain $\gamma=\sqrt{\mu\mu_{1}}$ and an $exp-KL$ function with respect to $m$.\label{208}
\end{proposition}
\begin{proof}
The proof of Proposition \ref{234} can be obtained from the proof of Theorem 1 in \cite{Sundaram2016} by letting $H=I$, $B=0$, $D=I$, $D_{q}=0$ and adding a new variable $\mu_{1}$ in $\Xi_{22}$.
\end{proof}

From Proposition \ref{208} , we see that if we could solve (\ref{29}) while minimizing $\sqrt{\mu\mu_{1}}$, then the designed observer is robust to measurement noise. We take advantage of the results in \cite{Sundaram2016} by using an incremental multiplier matrix to characterize the nonlinearity $f$ in the design of a robust circle-criterion observer, but we do not fix $\mu_{1}=1$ as a constant as \cite{Sundaram2016} does. Hence, our observer could provide estimates more robust to measurement noise in some circumstances.

From (\ref{14}), we have
\begin{equation}
(\hat{q}-\tilde{q})^{\top}\left( f\left( \hat{q})-f(\tilde{q}\right) \right) \geq 0.
\end{equation} 
Recalling (\ref{9}) and (\ref{19}), we know $\triangle q^{\top}\triangle f\geq 0$ $\forall\tilde{q}\in\mathbb{R}^{r}$ and $\forall\hat{q}\in\mathbb{R}^{r}$.  Hence, any matrix
\begin{equation}
M=\kappa\left[ \begin{matrix}
0&1\\
1&0
\end{matrix}\right]\label{72},
\end{equation}
with $\kappa>0$ is an incremental multiplier matrix for $f$. The following linear matrix inequality is equivalent to (\ref{29}). 

\begin{lemma}\cite{Sundaram2016} For some matrix $Y_{2}\in\mathbb{R}^{r\times n_{y}}$, consider the linear matrix inequality
	\begin{eqnarray}
	\left[ \begin{matrix}
	-P&\star\\
	\Xi_{21}&\Xi_{22}
	\end{matrix}\right]+\left[\begin{matrix}
	0&0\\
	0&\Gamma_{1}^{T}M\Gamma_{1}+\Gamma_{1}^{T}\Gamma_{2}+\Gamma_{2}^{T}\Gamma_{1}
	\end{matrix} \right] \leq 0\label{44},
	\end{eqnarray} 
	where $\Xi_{21}$, $\Xi_{22}$ are described in (\ref{225}), and 
	$
	\Gamma_{1}=\left[ \begin{matrix}H&0&0\label{81}\\
	0&0&I\end{matrix}\right],
	\Gamma_{2}=\left[ \begin{matrix}0&0&0\\
	Y_{2}C&-Y_{2}&0\end{matrix}\right]\label{82}
	$, then with
	\begin{equation}
L=P^{-1}Y,K=\frac{Y_{2}}{\kappa}\label{48}
\end{equation}
	(\ref{44}) and (\ref{29}) are equivalent.\label{209}
\end{lemma}
\begin{proof} Recalling (\ref{39}), we let $\Gamma=\Gamma_{1}+\tilde{\Gamma}_{2}$
	where $\tilde{\Gamma}_{2}=\left[ \begin{matrix}
	KC&-K&0\\
	0&0&0
	\end{matrix}\right] $. Note that $M\tilde{\Gamma}_{2}=\left[ \begin{matrix}
	0&0&0\\
	\kappa KC&\-\kappa K&0\end{matrix}\right] 
	$, with $K$ given by (\ref{48}), we see that $\kappa K=Y_{2}$. Therefore, $M\tilde{\Gamma}_{2}=\left[ \begin{matrix}
	0&0&0\\
	Y_{2}C&\-Y_{2}&0\end{matrix}\right]=\Gamma_{2} 
	$. As we have $\tilde{\Gamma}_{2}^{\top}M\tilde{\Gamma}_{2}=0$, thus
	\begin{eqnarray}
	\begin{split}
		\Gamma^{\top}M\Gamma=&(\Gamma_{1}+\tilde{\Gamma}_{2})^{\top}M(\Gamma_{1}+\tilde{\Gamma}_{2})\\
	=&\Gamma_{1}^{\top}M\Gamma_{1}+\Gamma_{1}^{\top}M\tilde{\Gamma}_{2}+\tilde{\Gamma}_{2}^{\top}M\Gamma_{1}+\tilde{\Gamma}_{2}^{\top}M\tilde{\Gamma}_{2}\\
	=&\Gamma_{1}^{\top}M\Gamma_{1}+\Gamma_{1}^{\top}\Gamma_{2}+\Gamma_{2}^{\top}\Gamma_{1}
	\end{split}
	\end{eqnarray}
	which implies that (\ref{44}) and (\ref{29}) are equivalent. The proof is complete.
\end{proof}
By replacing (\ref{29}) with (\ref{44}) in Proposition \ref{234}, we obtain the following result.

\begin{theorem}\label{240} Consider the system (\ref{13}), for given $c_{3}\in(0,1)$, suppose there exist matrix $P\in\mathbb{R}^{n\times n}$ and $P>0$, matrix $Y\in\mathbb{R}^{n\times n_{y}}$, matrix $Y_{2}\in\mathbb{R}^{r\times n_{y}}$, scalars $\kappa>0$, $\mu>0$, $\mu_{1}>0$ that satisfy the linear matrix inequalities
	\begin{eqnarray}
	\begin{split}
	\left[ \begin{matrix}
	-P&\star\\
	\Xi_{21}&\Xi_{22}
	\end{matrix}\right]+\left[\begin{matrix}
	0&0\\
	0&\Gamma_{1}^{T}M\Gamma_{1}+\Gamma_{1}^{T}\Gamma_{2}+\Gamma_{2}^{T}\Gamma_{1}
	\end{matrix} \right] \leq& 0,\label{68}\\
	\left[ \begin{matrix}P&I\\
	I&\mu I\end{matrix}\right] \geq& 0,
	\end{split}
	\end{eqnarray} 
	where $\Xi_{21}$, $\Xi_{22}$ are described in (\ref{225}), $M$ is given by  (\ref{72}), and $
	\Gamma_{1}=\left[ \begin{matrix}H&0&0\\
	0&0&I\end{matrix}\right],
	\Gamma_{2}=\left[ \begin{matrix}0&0&0\\
	Y_{2}C&-Y_{2}&0\end{matrix}\right]\label{84}
	$, then the observer (\ref{97}) characterized by gains given in (\ref{48})
has ISS error dynamics with a linear gain $\gamma=\sqrt{\mu\mu_{1}}$ and an $exp-KL$ function, which means there exist $c>0$, $\lambda\in(0,1)$ such that 
	\begin{equation}
	|e(k)|\leq c\lambda^{k}|e(0)|+\gamma||m||_{k},\label{94}
	\end{equation}
	for all $e(0)\in\mathbb{R}^{n}$, $k\geq 0$ and $m\in\mathbb{R}^{n_{y}}$ with $\left\lbrace m(k)\right\rbrace \in\l_{\infty}$.
\end{theorem}

\begin{corollary}
	A circle-criterion observer robust to measurement noise can be obtained by solving (\ref{68}) while minimizing $\mu+\mu_{1}$.
\end{corollary}

Theorem \ref{240} provides a way to design the observer for (\ref{13}). If we solve (\ref{68}) while minimizing $\sqrt{\mu\mu_{1}}$, we obtain an observer robust to measurement noise. To make the objective function convex, we consider using $\mu+\mu_{1}$ instead as our objective function. We know that $
(\mu+\mu_{1})^{2}\geq 4\cdot\mu\mu_{1}$, which yields
$\gamma=\sqrt{\mu\mu_{1}}\leq \frac{1}{2}(\mu+\mu_{1})$ as $\mu,\mu_{1}$ are positive. Therefore, we can minimize the upper bounds of $\gamma$ by minimizing $\mu+\mu_{1}$. By solving (\ref{68}) while minimizing $\mu+\mu_{1}$, we can obtain an observer that {attenuates} measurement noise. Since $c_{3}\in(0,1)$ is in a bounded set, we do a grid-search over $c_{3}$, i.e. we make a grid in $(0,1)$  and for each grid point we solve (\ref{68}) while minimizing $\mu+\mu_{1}$ and then we choose the $c_{3}$ that minimizes $\sqrt{\mu\mu_{1}}$. In our design method, besides regarding $\mu_{1}$ as a variable, we also do not assume $c_{3}$ is a fixed constant with a given value as \cite{Sundaram2016} does, which makes our LMIs less conservative than those in \cite{Sundaram2016} and further improves the robustness of our observer to measurement noise in some circumstances. We use the model used in Example 1 in \cite{Ibrir2007} and compare their performance by introducing measurement noise $m$. All LMIs were solved using PENLAB \cite{Fiala} in MATLAB.

\textbf{Example 1} Consider the discrete-time nonlinear system subject to measurement noise:
\begin{eqnarray}
\begin{split}
x^{+}=&\left[ \begin{matrix}
1&\delta\\
0&1
\end{matrix}\right]x+\left[ \begin{matrix}
\frac{1}{2}\delta\alpha \sin (x_{1}+x_{2})\\
\delta\alpha \sin (x_{1}+x_{2})
\end{matrix}\right]
+\left[ \begin{matrix}
\delta u\\
\delta u
\end{matrix}\right]\label{d}, \\
y=&\left[ \begin{matrix}
3&0.3\\
3&0.6\\
6&0.9\\
1.2&12
\end{matrix}\right] x+m.
\end{split}
\end{eqnarray}
We let $\delta=0.1$ and $\alpha=1$. (\ref{d}) can be rewritten in the form of (\ref{13}) with Assumption \ref{1000} holding, see \cite{Ibrir2007} for more details. We solve (\ref{68}) while minimizing $\mu+\mu_{1}$ and doing a grid search over $c_{3}$ to obtain observer matrices $K$, $L$, $c_{3}=0.900$, and $\gamma=0.924$. We obtain $K$, $L$ via solving the LMIs in \cite{Ibrir2007}. We solve the LMIs in \cite{Sundaram2016} by letting $c_{3}=0.500$, $M$ given as (\ref{72}), $H=I$ and minimizing $\mu$ to obtain observer with $\gamma=22.4$. We let $m\sim\mathcal{U}(-0.5,0.5)$. $x(0)$ is randomly selected from a normal distribution and $\hat{x}(0)=0 $. The performance of these observers are compared respectively in Figures \ref{fig:fig1}-\ref{fig:fig2}. 
\begin{figure}[h]
	\includegraphics[width=0.5\textwidth]{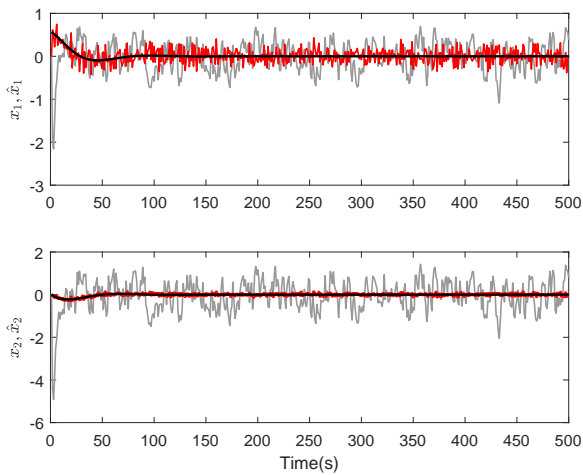}
	\caption{Estimated states $\hat{x}$ converges to a neighbourhood of the true states $x$. Legend: Observer obtained via Theorem 1 (red), Observer from \cite{Sundaram2016} (grey), true states (black)}
	\label{fig:fig1}
	\centering
\end{figure}
\begin{figure}[h]
	\includegraphics[width=0.5\textwidth]{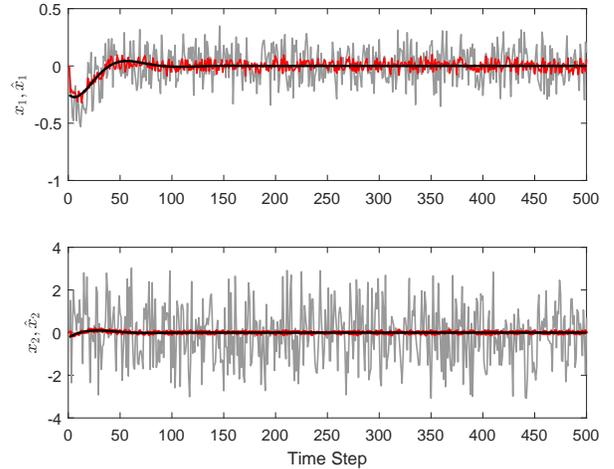}
	\caption{Estimated states $\hat{x}$ converges to a neighbourhood of the true states $x$. Legend: Observer obtained via Theorem 1 (red), Observer from \cite{Ibrir2007} (grey), true states (black)}
	\label{fig:fig2}
	\centering
\end{figure}
\section{A circle-criterion observer-based estimator robust to measurement noise and sensor attacks}

In this section, we introduce a circle-criterion observer-based estimator for the same class of discrete-time nonlinear systems as in Section \uppercase\expandafter{\romannumeral3}, but we assume a small number of sensors are also subject to sensor attacks:
\begin{eqnarray}
\begin{split}
x^{+}=&Ax+Gf(Hx)+\rho(u,y)\label{3},\\
\tilde{y}=&\tilde{C}x+a+\tilde{m},
\end{split}
\end{eqnarray}
where $x\in\mathbb{R}^{n}$ is the state, $\tilde{y}\in\mathbb{R}^{p}$ is the sensor measurement, $\tilde{m}\in\mathbb{R}^{p}$ is the measurement noise with $\left\lbrace \tilde{m}(k)\right\rbrace \in l_{\infty}$, $a\in\mathbb{R}^{p}$ is the vector of attacks: if sensor $i\in\left\lbrace 1,\cdots,p\right\rbrace $ is not attacked, then the $ith$ component of the vector $a(k)$, $a_{i}(k)=0, \forall k\geq 0$; otherwise sensor $i$ is attacked and { $a_{i}(k)$ is arbitrary and possibly unbounded}. We denote $W\subseteq\left\lbrace 1,\cdots, p\right\rbrace $ the set of attacked sensors, then we have $\supp(a) = W$. We assume the set $W$ is unknown to us. We denote $\tilde{y}\left( k;x(0),a_{[0,k]},\tilde{m}_{[0,k]}\right)$ as the output of the system at time $k$ when the initial state is $x(0)$ and the outputs are subject to measurement noise $\tilde{m}$ and sensor attacks $a$. (\ref{230}), (\ref{14}) still hold.

Suppose some of the sensors are corrupted by attack signals. Having received the measured output sequences $\left\lbrace \tilde{y}(k;x(0),a_{[0,k]},\tilde{m}_{[0,k]})\right\rbrace $ $\forall k\geq0$, a circle-criterion observer-based estimator is used to estimate the states of the system. Our objective is to present a design method of a robust circle-criterion observer-based estimator that provides exponential convergence of the estimates $\hat{x}(k)$ to a neighborhood of the true states $x(k)$ under some assumptions, and the error $e(k)=\hat{x}(k)-x(k)$ is ISS with a linear gain and $exp-KL$ function with respect to the measurement noise only, which means there exist $\bar{c}>0$, $\bar{\lambda}\in(0,1)$, $\bar{\gamma}_{y}\geq0$ such that,
\begin{equation}
|e(k)|\leq\bar{c}\bar{\lambda}^{k}|e(0)|+\bar{\gamma}_{y}||\tilde{m}||_{k},
\end{equation}
for all $e(0)\in\mathbb{R}^{n}$, $k\geq 0$.

We now outline our estimation strategy which is inspired by the method for the linear case from \cite{Chong2015} . For (\ref{3}), let $0<q<\frac{p}{2}$ be the largest integer such that for each subset $J\subset\left\lbrace 1,\cdots, p\right\rbrace $ of sensors with $\card(J)\geq p-2q$, the circle-criterion observer of the form
\begin{eqnarray}
\begin{split}
\hat{x}_{J}^{+}=&A\hat{x}_{J}+Gf(H\hat{x}_{J}+K_{J}(\tilde{C}_{J}\hat{x}_{J}-\tilde{y}_{J}))\\
&+L_{J}(\tilde{C}_{J}\hat{x}_{J}-\tilde{y}_{J})+\rho(u,y)\label{88},
\end{split}
\end{eqnarray}
exists for $\tilde{y}_{J}$. $\hat{x}_{J}$ denotes the estimate of the state $x$ from $\tilde{y}_{J}$, and $K_{J}$, $L_{J}$ are the observer gains. $\tilde{C}_{J}$ is the stacking of all $\tilde{C}_{i}$, $i\in J$ where $\tilde{C}_{i}$ is the $ith$ row of $\tilde{C}$. 

When we say the observer exists for $\tilde{y}_{J}$, we mean when $a_{J}(k)=0$ for all $k\geq 0$ the error of each observer $e_{J}(k)=\hat{x}_{J}(k)-x(k)$ with the following dynamics
\begin{equation}
e_{J}^{+}=(A+L_{J}\tilde{C}_{J})e_{J}-L_{J}\tilde{m}_{J}+G\triangle f_{J}\label{89},
\end{equation}
where $\triangle f_{J}=f(\hat{q})-f(\tilde{q}_{J})$, $\tilde{q}=Hx$ and $\hat{q}_{J}=H\hat{x}_{J}+K_{J}(\hat{y}_{J}-\tilde{y}_{J})$ and $\hat{y}_{J}=\tilde{C}_{J}\hat{x}_{J}$, $\tilde{y}_{J}=\tilde{C}_{J}x+\tilde{m}_{J}$
is ISS with a linear gain $\gamma_{J}$ and an $exp-KL$ function with respect to measurement noise $\tilde{m}_{J}$. This implies that there exist $c_{J}>0$, $\lambda_{J}\in(0,1)$, $\gamma_{J}\geq0$ such that 
\begin{equation}
|e_{J}(k)|\leq c_{J}\lambda^{k}_{J}|e_{J}(0)|+\gamma_{J}||\tilde{m}_{J}||_{k},\label{233}
\end{equation}
for all $e_{J}(0)\in\mathbb{R}^{n}$, $k\geq 0$ and $ \tilde{m}_{J}\in\mathbb{R}^{\card(J)}$ with $\left\lbrace \tilde{m}_{J}(k)\right\rbrace \in l_{\infty}$. 

We assume that:
\vspace{2mm}
\begin{assumption} \label{1001} There are at most $q$ sensors attacked,
	\begin{equation}
	\card(W)\leq q.
	\end{equation}
\end{assumption}
\vspace{2mm}
By using the design method proposed in Section \uppercase\expandafter{\romannumeral3}, we construct a robust observer for each subset $J\subset\left\lbrace 1,2,\cdots,p\right\rbrace $ with $\card(J)=p-q$ and for each subset $S\subset \left\lbrace 1,\cdots,p\right\rbrace $ with $\card(S)=p-2q$. For each subset $J$ with $\card(J)=p-q$, we define $\pi_{J}(k)$ for all $k\geq0$ to be the largest deviation between the estimate $\hat{x}_{J}(k)$ and the estimate $\hat{x}_{S}(k)$ that is given by any subset $S\subset J$ with $\card(S)=p-2q$.
\begin{equation} \label{54}
\pi_{J}(k):=\max_{S\subset J:\card(S)=p-2q}|\hat{x}_{J}(k)-\hat{x}_{S}(k)|.
\end{equation}
Recalling that among the total $p$ sensors, there is at least one subset $\bar{I}\subset\left\lbrace 1,\cdots, p\right\rbrace $ of sensors with $\card(\bar{I})=p-q$ that $\tilde{y}_{\bar{I}}=C_{\bar{I}}x+\tilde{m}_{\bar{I}}$ as $a_{\bar{I}}=0$, then in general all of the estimates that appear in the definition of $\pi_{\bar{I}}(k)$ are more consistent than all the subsets $J$ with $\card(J)=p-q$ and $\tilde{y}_{J}=\tilde{C}_{J}x+a_{J}+\tilde{m}_{J}$ with $a_{J}\neq 0$. This motivates the following state estimation: for all $k\geq 0$,
\begin{equation}\label{55}
\sigma(k)=\underset{J\subset\left\lbrace 1,2,\cdots,p\right\rbrace :\card(J)=p-q}{\argmin} \pi_{J}(k),
\end{equation}
and then we say for all $k\geq 0$, the estimate given by the subset $\sigma(k)$ is a good estimate, 
\begin{equation}\label{56}
\hat{x}(k)=\hat{x}_{\sigma(k)}(k),
\end{equation}
where $\hat{x}_{\sigma(k)}(k)$ represents the estimates given by the subset $\sigma(k)$. The following result states that the proposed estimator is robust with respect to sensor attacks and measurement noise. For simplicity, we initialize all the observers to the same condition $\hat{x}(0)$.

\begin{theorem} For the system (\ref{3}), suppose Assumptions \ref{1000}-\ref{1001} hold, recalling from (\ref{54})-(\ref{56}), denoting $e(k)=\hat{x}(k)-x(k)$, there exist positive constants $\bar{c}>0$, $\bar{\lambda}\in(0,1)$, $\bar{\gamma}_{y}\geq0$ such that the following inequality holds:
	\begin{eqnarray}
	|e(k)|\leq\bar{c}\bar{\lambda}^{k}|e(0)|+\bar{\gamma}_{y}||\tilde{m}||_{k} \label{60},
	\end{eqnarray}
	for all $e(0)\in\mathbb{R}^{n}$, $k\geq 0$, and $\tilde{m}\in\mathbb{R}^{p}$ with $\left\lbrace \tilde{m}(k)\right\rbrace \in l_{\infty}$.
\end{theorem}

\begin{proof} From the result of Section \uppercase\expandafter{\romannumeral3}, we know for each subset $J\subset\left\lbrace 1,\cdots, p\right\rbrace $ with $\card(J)\geq p-2q$, the observation error dynamics satisfies (\ref{233}). Since $a_{i}(k)=0$ for all $i\in\left\lbrace 1,\cdots,p\right\rbrace \backslash \supp(a)$ and $\forall k\geq 0$, we conclude for $J=\bar{I}\subseteq\left\lbrace 1,\cdots,p\right\rbrace \backslash \supp(a)$ with $\card(\bar{I})=p-q$, there exist $c_{\bar{I}}>0$, $\lambda_{\bar{I}}\in(0,1) $ and $\gamma_{\bar{I}}\geq0$, such that
	\begin{equation}\label{63}
	|e_{\bar{I}}(k)|\leq c_{\bar{I}}\lambda_{\bar{I}}^{k}|e(0)|+\gamma_{\bar{I}}||\tilde{m}_{\bar{I}}||_{k},
	\end{equation}
	for all $e(0)\in\mathbb{R}^{n}$ and $k\geq0$. Also for any set $S\subset\bar{I}$ with $\card(S)=p-2q$, we have $a_{S}(k)=0$ $\forall k\geq0$, hence there exist $c_{S}>0$, $\lambda_{S}\in(0,1)$ and $\gamma_{S}\geq0$ such that
	\begin{equation}\label{61}
	|e_{S}(k)|\leq c_{S}\lambda_{S}^{k}|e(0)|+\gamma_{S}||\tilde{m}_{S}||_{k},
	\end{equation}
	for all $e(0)\in\mathbb{R}^{n}$ and $k\geq 0$. Recalling the definition of $\pi_{\bar{I}}$ from (\ref{54}), we have that
	\begin{eqnarray}
	\begin{split}
	\pi_{\bar{I}}(k)=&\underset{S\subset\bar{I}}{\max}|\hat{x}_{\bar{I}}(k)-\hat{x}_{S}(k)|\\
	=&\underset{S\subset\bar{I}}{\max}|\hat{x}_{\bar{I}}(k)-x(k)+x(k)-\hat{x}_{S}(k)|\\
	\leq&  |e_{\bar{I}}(k)|+\underset{S\subset\bar{I}}{\max}|e_{S}(k)|
	\end{split}
	\end{eqnarray}
	for all $k\geq 0$. From (\ref{63}) and (\ref{61}), we obtain
	\begin{equation}\label{66}
	\pi_{\bar{I}}(k)\leq 2c'_{\bar{I}}\lambda_{\bar{I}}^{'k}|e(0)|+2\gamma'_{\bar{I}}||\tilde{m}_{\bar{I}}||_{k},
	\end{equation}
	for all $e(0)\in\mathbb{R}^{n}$ and $k\geq 0$, where $c'_{\bar{I}}:=\underset{S\subset\bar{I}}{\max}\left\lbrace c_{\bar{I}}, c_{S}\right\rbrace $, $\lambda'_{\bar{I}}:=\underset{S\subset\bar{I}}{\max}\left\lbrace \lambda_{\bar{I}}, \lambda_{S}\right\rbrace $, and $\gamma'_{\bar{I}}:=\underset{S\subset\bar{I}}{\max}\left\lbrace \gamma_{\bar{I}}, \gamma_{S}\right\rbrace$. Observe that since $S\subset\bar{I}$ with $\card(S)=p-2q$. Recall from (\ref{54})-(\ref{56}) that $\hat{x}(k)=\hat{x}_{\sigma(k)}(k)$ where $\sigma(k)=\underset{J\subset\left\lbrace 1,2,\cdots,p\right\rbrace :\card(J)=p-q}{\argmin} \pi_{J}(k)$, hence $\pi_{\sigma(k)}(k)\leq\pi_{\bar{I}}(k)$. We know that there exist at least one set $\bar{S}\subset\sigma(k)$ with $\card(\bar{S})=p-2q$ such that $a_{\bar{S}}(k)=0$ $\forall k\geq 0$, and there exist $c_{\bar{S}}>0$, $\lambda_{\bar{S}}\in(0,1)$ and $\gamma_{\bar{S}}\geq0$ such that 
	\begin{equation}\label{67}
	|e_{\bar{S}}(k)|\leq c_{\bar{S}}\lambda_{\bar{S}}^{k}|e(0)|+\gamma_{\bar{S}}||\tilde{m}_{\bar{S}}||_{k},
	\end{equation}
	for all $e(0)\in\mathbb{R}^{n}$ and $k\geq0$. From (\ref{54}), there is a fact that $\pi_{\sigma(k)}(k)=\underset{S\subset \sigma(k):\card(S)=p-2q}{\max}|\hat{x}_{\sigma(k)}(k)-\hat{x}_{S}(k)|\geq|\hat{x}_{\sigma(k)}(k)-\hat{x}_{\bar{S}}(k)|$. From the triangle inequality we have that 
	\begin{eqnarray}
	\begin{split}
		|e_{\sigma(k)}(k)|=&|\hat{x}_{\sigma(k)}(k)-x(k)|\\
	=&|\hat{x}_{\sigma(k)}(k)-\hat{x}_{\bar{S}}(k)+\hat{x}_{\bar{S}}(k)-x(k)|\\
	\leq&|\hat{x}_{\sigma(k)}(k)-\hat{x}_{\bar{S}}(k)|+|e_{\bar{S}}(k)|\\
	\leq&\pi_{\sigma(k)}(k)+|e_{\bar{S}}(k)|\\
	\leq&\pi_{\bar{I}}(k)+|e_{\bar{S}}(k)|
	\end{split}
	\end{eqnarray}
	for all $k\geq 0$. From (\ref{66}) and (\ref{67}), we have 
	\begin{eqnarray}\label{70}
	|e_{\sigma(k)}(k)|\leq \bar{c}\bar{\lambda}^{k}|e(0)|+\bar{\gamma}_{y}\cdot \max\left\lbrace ||\tilde{m}_{\bar{S}}||_{k},||\tilde{m}_{\bar{I}}||_{k}\right\rbrace,
	\end{eqnarray}
	for all $e(0)\in\mathbb{R}^{n}$ and $k\geq 0$, where $\bar{c}=3\cdot \max\left\lbrace c_{\bar{S}},c'_{\bar{I}}\right\rbrace $, $\bar{\lambda}=\max\left\lbrace \lambda_{\bar{S}},\lambda'_{\bar{I}}\right\rbrace $, $\bar{\gamma}_{y}=3\cdot \max\left\lbrace \gamma_{\bar{S}}, \gamma'_{\bar{I}}\right\rbrace $. Since $||\tilde{m}||_{k}\geq \max\left\lbrace ||\tilde{m}_{\bar{S}}||_{k},||\tilde{m}_{\bar{I}}||_{k}\right\rbrace$, we can see (\ref{70}) satisfies (\ref{60}). The proof is complete.
\end{proof}


 We still use the model in Example 1, but here we assume sensor attacks and measurement noise both occur to test the performance of our designed estimator.

\textbf{Example 2} Consider the discrete-time nonlinear system subject to measurement noise and sensor attacks:
\begin{eqnarray}
\begin{split}
x^{+}=&\left[ \begin{matrix}\label{e1}
1&\delta\\
0&1
\end{matrix}\right]x+\left[ \begin{matrix}
\frac{1}{2}\delta\alpha \sin (x_{1}+x_{2})\\
\delta\alpha\sin (x_{1}+x_{2})
\end{matrix}\right] 
+\left[ \begin{matrix}
\delta u\\
\delta u
\end{matrix}\right], \\
\tilde{y}=&\left[ \begin{matrix}
3&0.3\\
3&0.6\\
6&0.9\\
1.2&12
\end{matrix}\right] x+a+\tilde{m}.
\end{split}
\end{eqnarray}

We still let $\delta=0.1$ and $\alpha=1$, $\tilde{m}\sim\mathcal{U}(-0.5,0.5)$. We find that the circle-criterion observer of the form (\ref{88}) exists for each subset of $J\subset\left\lbrace 1,2,3,4\right\rbrace $ with $\card(J)\geq 1$ and $p=4$, we have $q=1$. We let $W=\left\lbrace 3\right\rbrace $, which means the $3$-rd sensor is under attack. The estimator knows there is at most one sensor under attack, but does not know which. By using the design method proposed in Section \uppercase\expandafter{\romannumeral3}, we design an observer for each $J\subset\left\lbrace 1,2,3,4\right\rbrace $ with $\card(J)=3$ and each $S\subset\left\lbrace 1,2,3,4\right\rbrace $ with $\card(S)=2$. Therefore, totally  $\binom{4}{3}+\binom{4}{2}=10$ observers are designed, and they are all initialized at $\hat{x}(0)=\left[ 0,0\right]^{\top} $. $x_{1}(0),x_{2}(0)$ are randomly selected from a standard normal distribution. We let $a_{3}\sim\mathcal{U}(-b,b)$ with $b$ given by $1,10$. For all $k\in[0,500]$, (\ref{54})-(\ref{56}) is used to construct $\hat{x}(k)$. The performance of the designed estimator is shown in Figures \ref{fig:3e}-\ref{fig:4e}.
%
\begin{figure}[h]
	\includegraphics[width=0.5\textwidth]{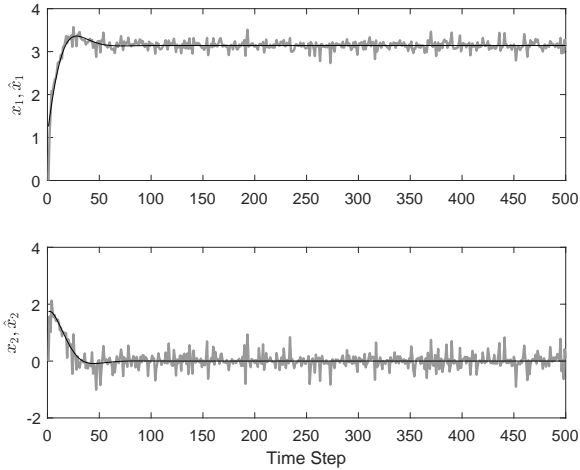}
	\caption{Estimated states $\hat{x}$ converges to a neighbourhood of the true states $x$ when $a_{3}\sim\mathcal{U}(-1,1)$. Legend: $\hat{x}$ (grey), true states (black)}
	\label{fig:3e}
	\centering
\end{figure}
\begin{figure}[h]
	\includegraphics[width=0.5\textwidth]{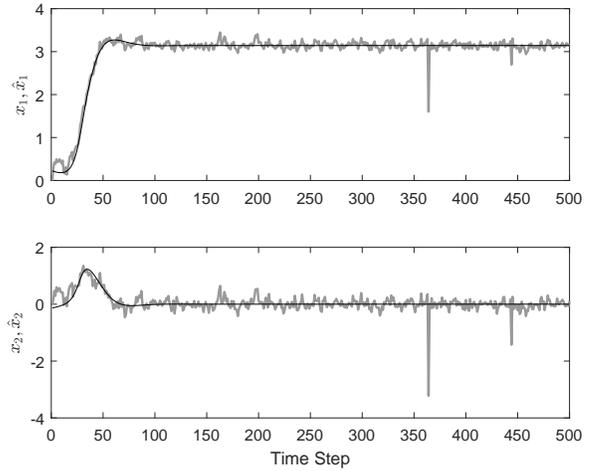}
	\caption{Estimated states $\hat{x}$ converges to a neighbourhood of the true states $x$ when $a_{3}\sim\mathcal{U}(-10,10)$. Legend: $\hat{x}$ (grey), true states (black)}
	\label{fig:4e}
	\centering
\end{figure}
\section{Isolation of sensor attacks}
In this section, we still consider system (\ref{3}). We let $q$ be the largest integer such that a circle-criterion observer exists for each subset $J\subset\left\lbrace 1,\cdots,p\right\rbrace $ with $\card(J)\geq p-2q$. We propose an algorithm for isolating attacked sensors when we know how many sensors are attacked, which is denoted as $q^{\star}$ ($q^{\star}\leq q$).
\vspace{2mm}
\begin{assumption}\label{cc} There are $q^{\star}\leq q$ attacked sensors, i.e.,
	\begin{equation}
		\card(W)=q^{\star},
	\end{equation}
	and $q^{\star}\leq q$ is a known positive integer.
\end{assumption}
\vspace{2mm}
 We construct a circle-criterion observer for each subset $J\subset\left\lbrace 1,\cdots,p\right) $ with $\card(J)=p-q^{\star}$ and for each subset $S\subset\left\lbrace 1,\cdots,p\right\rbrace $ with $\card(S)=p-2q^{\star}$. For each subset $J$ with $\card(J)=p-q^{\star}$ and for all $k\geq0$, we define $\pi_{J}^{\star}(k)$ as
\begin{equation}
	\pi_{J}^{\star}(k):=\max_{S\subset J:\card(S)=p-2q^{\star}}|\hat{x}_{J}(k)-\hat{x}_{S}(k)|.
\end{equation}
 Since there are $q^{\star}$ sensors under attack, we know there is one subset $\bar{I}\subset\left\lbrace 1,\cdots, p\right\rbrace $ of sensors with $\card(\bar{I})=p-q^{\star}$ that $\tilde{y}_{\bar{I}}=\tilde{C}_{\bar{I}}x+\tilde{m}_{\bar{I}}$ as $a_{\bar{I}}=0$, then all of the estimates that appear in the definition of $\pi_{\bar{I}}(k)$ are very likely to be more consistent than all the subsets $J$ with $\card(J)=p-q^{\star}$ and $\tilde{y}_{J}=\tilde{C}_{J}x+a_{J}+\tilde{m}_{J}$ with $a_{J}\neq 0$. For all $k>0$, if we denote $\bar{J}(k)$ as the set of attack-free sensors at time $k$, then $\bar{J}(k)$ is given as
\begin{equation}\label{a}
	\bar{J}(k)=\underset{J\subset\left\lbrace 1,2,\cdots,p\right\rbrace :\card(J)=p-q^{\star}}{\argmin} \pi_{J}(k).
\end{equation}
Then the set $\left\lbrace 1,\cdots,p\right\rbrace \setminus\bar{J}(k)$ is isolated as the set of attacked sensors at time $k$. We make our decision in every $N$ time steps, where $N\in\mathbb{Z}_{>0}$ is the window size we choose, i.e. in each $N$ time steps we keep obtaining $\bar{J}(k)$ from (\ref{a}) for each $k$, and we choose the subset $J(i)$ that is equal to $\bar{J}(k)$ most often in the $i$-th window. Then we claim $\left\lbrace 1,\cdots,p\right\rbrace \setminus J(i)$ is the set of sensors potentially under attack in the $i$-th time window, which is Algorithm \ref{alg:the_alg2}.
 \begin{algorithm}
	\caption{SENSOR ATTACKS ISOLATION \textbf{Input}: $N$, $q^{\star}$ }
	\label{alg:the_alg2}
	\begin{algorithmic}[1]
		\State \text{Design a circle-criterion observer for each subset} \text{$J\subset\left\lbrace 1,\cdots,p\right\rbrace $ with $\card(J)=p-q^{\star}$ and for each} \text{subset $S\subset\left\lbrace 1,\cdots,p\right\rbrace $ with $\card(S)=p-2q^{\star}$ }. 
		\State \text{We intialize the counter variable $n_{J}(i)=0$ for each $J$} \text{with $\card(J)=p-q^{\star}$ and for all $ i\in\mathbb{Z}_{>0}$.}
		\For{each $i\in\mathbb{Z}_{>0}$} \For{each $k\in \left[ 1+(i-1)N,iN\right] $} 
		\State calculate $\pi_{J}^{\star}(k)$ for all $J$ with $\card(J)=p-q^{\star}$ as follows:
		\begin{equation*}
		\pi_{J}^{\star}(k)=\max_{S\subset J:\card(S)=p-2q^{\star}}|\hat{x}_{J}(k)-\hat{x}_{S}(k)|.
		\end{equation*}
		\State \text{Select the subset $\bar{J}(k)$} \text{such that}
		\begin{equation*}
		\bar{J}(k)=\underset{J\subset\left\lbrace 1,2,\cdots,p\right\rbrace :\card(J)=p-q^{\star}}{\argmin} \pi_{J}(k).
		\end{equation*}
		\If{for some $J$ with $\card(J)=p-q^{\star}$ we have $\bar{J}(k)==J$}
		\State update $n_{J}(i)$ as follows:
		\begin{equation*}
		n_{J}(i)=n_{J}(i)+1.
		\end{equation*}
		\EndIf
		\EndFor
		\State \text{Select the subset $J$ that is equal to $\bar{J}(k)$ most often} 
		\begin{equation*}
		J(i)=\underset{J\in\left\lbrace 1,\cdots,p\right\rbrace :\card(J)=p-q^{\star}}{\argmax} n_{J}(i).
		\end{equation*}
		\State \text{The set of sensors potentially under attack is given as:}
		\begin{equation*}
		\tilde{A}(i) = \left\lbrace 1,\cdots,p\right\rbrace \setminus J(i).
		\end{equation*}
		\State \textbf{Return} $\tilde{A}(i)$.
		\EndFor
	\end{algorithmic}
\end{algorithm}

\textbf{Example 3}
We still consider model (\ref{e1}) in Example 2,
with $\delta=0.1$. We consider two cases where $\alpha$ is equal to $1$ and $0$ respectively. In each case, we let $\tilde{m}\sim\mathcal{U}(-0.5,0.5)$, $q^{\star}=1$ and $W=\left\lbrace 3\right\rbrace $. We let $a_{3}\sim\mathcal{U}(-b,b)$, and $b$ given by $1,2.5$. In each case, we choose the window size $N$ to be $50,100,200$ respectively. \\
\textit{Case 1.} $\alpha=1$, we apply Algorithm 1 by running $\binom{4}{2} +\binom{4}{3}=10 $ circle-criterion observers which are all initialized with $\hat{x}(0)=\left[ 0,0\right] ^{\top}$ and $x_{1}(0),x_{2}(0)$ are randomly selected from a standard normal distribution. We follow the steps in Algorithm \ref{alg:the_alg2}. We check in $1000$ time steps which sensor is isolated in each time window, which is shown in Figures \ref{fig:8a}-\ref{fig:9a}. \\
\textit{Case 2.}  $\alpha=0$, (\ref{e1}) becomes a discrete-time linear systems subject to measurement noise and sensor attacks. We construct observers via solving (\ref{68}) by letting $G=0$ and minimizing $\mu+\mu_{1}$. We apply Algorithm \ref{alg:the_alg2} in a similar way as what we do when $\alpha=1$ and Figures \ref{fig:8d}-\ref{fig:9d} show the performance of Algorithm \ref{alg:the_alg2} when $\alpha=0$.

 \begin{figure}[h]
 	\includegraphics[width=0.5\textwidth]{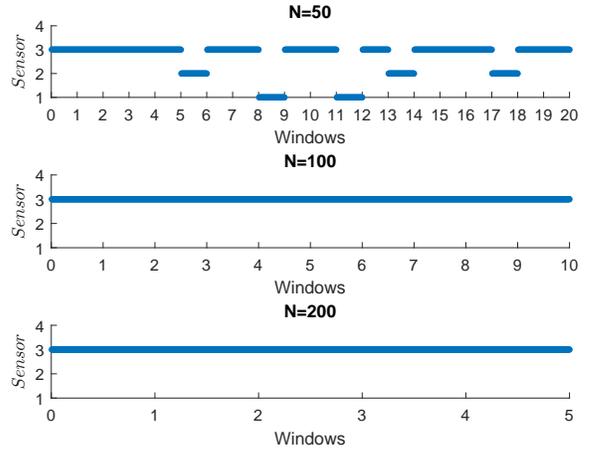}
 	\caption{The sensor isolated by Algorithm \ref{alg:the_alg2}, $\alpha=1$, $a_{3}\sim\mathcal{U}(-1,1)$.}
 	\label{fig:8a}
 	\centering
 \end{figure}
 \begin{figure}[h]
 	\includegraphics[width=0.5\textwidth]{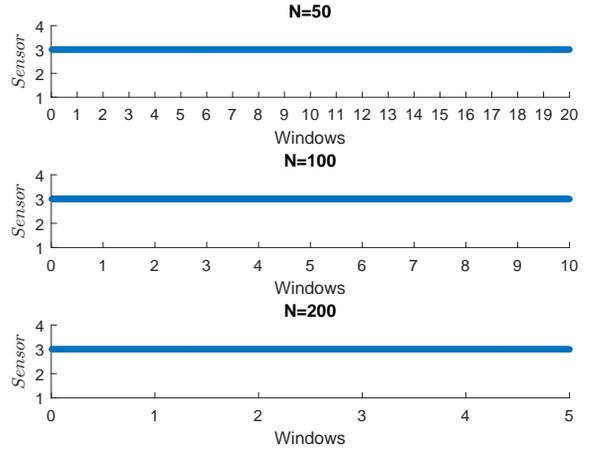}
 	\caption{The sensor isolated  by Algorithm \ref{alg:the_alg2}, $\alpha=1$, $a_{3}\sim\mathcal{U}(-2.5,2.5)$.}
 	\label{fig:9a}
 	\centering
 \end{figure}
The major advantage of Algorithm \ref{alg:the_alg2} is that it can be applied to isolate attacked sensors when sensor attacks and measurement noise both occur as long as measurement noise is bounded.
\begin{figure}[h]
	\includegraphics[width=0.5\textwidth]{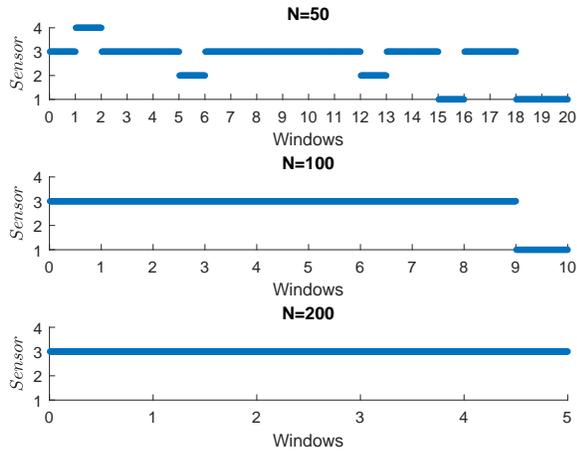}
	\caption{The sensor isolated by Algorithm \ref{alg:the_alg2}, $\alpha=0$, $a_{3}\sim\mathcal{U}(-1,1)$.}
	\label{fig:8d}
	\centering
\end{figure}
\begin{figure}[h]
	\includegraphics[width=0.5\textwidth]{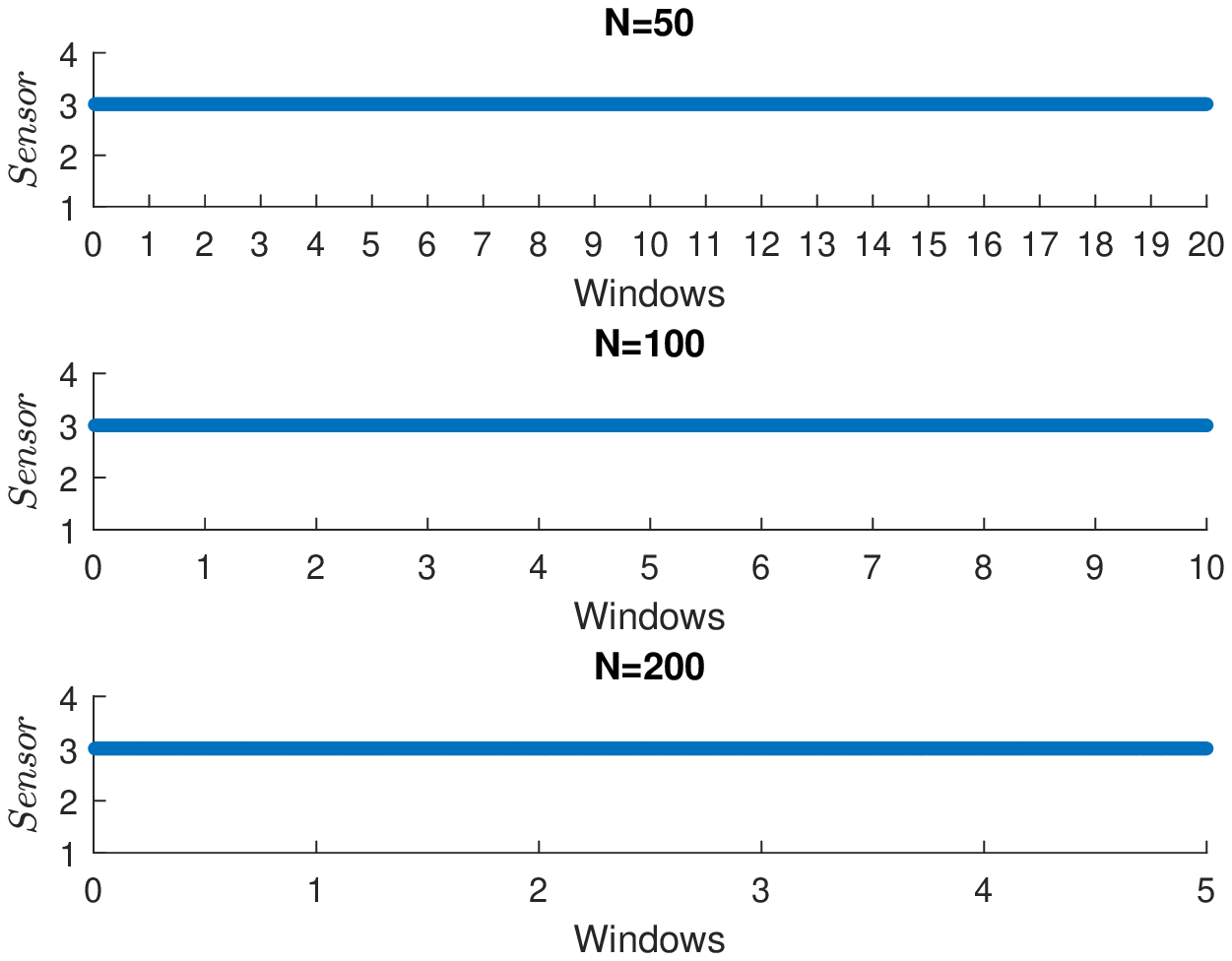}
	\caption{The sensor isolated by Algorithm \ref{alg:the_alg2}, $\alpha=0$, $a_{3}\sim\mathcal{U}(-2.5,2.5)$.}
	\label{fig:9d}
	\centering
\end{figure}
\section{Conclusion}
Following the way of \cite{Sundaram2016}, a design method of a discrete-time circle-criterion observer robust to measurement noise is given as a series of linear matrix inequalities in the absence of attack signals. An less conservative ISS gain is obtained by solving an optimization problem with more degrees of freedom. Then a circle-criterion observer-based estimation strategy is proposed in the presence of measurement noise and sensor attacks. We show that the designed circle-criterion observer-based estimator provides ISS estimation errors with a linear gain and an $exp-KL$ function with respect to measurement noise when a sufficiently small subset of sensors are corrupted by (potentially unbounded) attack signals and all sensors are affected by bounded measurement noise. This work can be seen as an extension of the existing observer-based estimator for linear systems \cite{Chong2015}. An algorithm for isolating attacked sensors is also proposed when we know how many sensors are attacked.
\bibliographystyle{ieeetr}
\bibliography{Observer} 
\end{document}